\documentclass[runningheads,letter, 11pts]{llncs}

\usepackage{url}
\newcommand{\keywords}[1]{\par\addvspace\baselineskip
\noindent\keywordname\enspace\ignorespaces#1}

\usepackage{amsfonts}
\usepackage{amssymb}
\usepackage{graphicx}
\usepackage{enumerate}
\usepackage{amsmath}

\newcommand{\beeq}[1]{\begin{equation} \label{#1}}
\newcommand{\eeq}{\end{equation}}
\newcommand{\beeqs}{\begin{eqnarray*}}
\newcommand{\eeqs}{\end{eqnarray*}}
\renewcommand{\(}{\begin{eqnarray*}}
\renewcommand{\)}{\end{eqnarray*}}
\newcommand{\beeqn}{\begin{eqnarray}}
\newcommand{\eeqn}{\end{eqnarray}}

\newcommand{\eqand}{\mbox{~~~and~~~}}
\newcommand{\eqwhere}{\mbox{~~~where~~~}}

\newcommand{\refth}[1]{Theorem~\ref{#1}}
\newcommand{\reflm}[1]{Lemma~\ref{#1}}

\newcommand{\refsec}[1]{Section~\ref{#1}}
\newcommand{\reffig}[1]{Fig.\ \ref{#1}}
\newcommand{\refeq}[1]{(\ref{#1})}

\renewcommand{\quad}{\hspace*{3mm}}
\renewcommand{\qquad}{\hspace*{5mm}}

\newcommand{\lp}{\left(  }
\newcommand{\rp}{\right) }
\newcommand{\lb}{\left\{  }
\newcommand{\rb}{\right\} }
\newcommand{\lbr}{\left[  }
\newcommand{\rbr}{\right] }
\newcommand{\lf}{\left\lfloor}
\newcommand{\rf}{\right\rfloor}
\newcommand{\lc}{\left\lceil}
\newcommand{\rc}{\right\rceil}

\newcommand{\Z}{{\mathbb Z}}
\newcommand{\N}{{\mathbb N}}

\newcommand{\R}{{\mathbb R}}

\newcommand{\ve}{\varepsilon}
\newcommand{\ep}{\epsilon}

\newcounter{cnt1}
\newcounter{cnt2}
\newcounter{cnt3}
\newcounter{cnt4}
\newcounter{cnt5}

\newcommand{\beenu}
{
\begin{list}{\arabic{cnt1}.}
{\usecounter{cnt1}
\leftmargin 4mm
\setlength{\leftmargin}{\leftmargin}
\topsep 0pt
\parsep 0pt
\itemsep 0pt}
}
\newcommand{\eenu}{\end{list}}

\newcommand{\beenub}
{
\begin{list}{\arabic{cnt1}-\arabic{cnt2}.}
{
\leftmargin  4mm
\setlength{\leftmargin}{\leftmargin}
\topsep 0pt
\parsep 0pt
\itemsep 0pt}
}
\newcommand{\eenub}{\end{list}}

\newcommand{\beenuc}
{
\begin{list}{\arabic{cnt1}-\arabic{cnt2}-\arabic{cnt3}.}
{
\leftmargin  4mm
\setlength{\leftmargin}{\leftmargin}
\topsep 0pt
\parsep 0pt
\itemsep 0pt}
}
\newcommand{\eenuc}{\end{list}}

\newcommand{\beenud}
{
\begin{list}{\arabic{cnt1}-\arabic{cnt2}-\arabic{cnt3}-\arabic{cnt4}.}
{
\leftmargin  4mm
\setlength{\leftmargin}{\leftmargin}
\topsep 0pt
\parsep 0pt
\itemsep 0pt}
}
\newcommand{\eenud}{\end{list}}

\newcommand{\beenue}
{
\begin{list}{\arabic{cnt1}-\arabic{cnt2}-\arabic{cnt3}-\arabic{cnt4}-\arabic{cnt5}.}
{
\leftmargin  4mm
\setlength{\leftmargin}{\leftmargin}
\topsep 0pt
\parsep 0pt
\itemsep 0pt}
}
\newcommand{\eenue}{\end{list}}

\newcommand{\broman}
{
\begin{list}{\roman{cnt1})}
{
\usecounter{cnt1}
\leftmargin 3mm
\setlength{\leftmargin}{\leftmargin}
\topsep 0pt
\parsep 0pt
\itemsep 0pt}
}
\newcommand{\eroman}{\end{list}}

\newcommand{\bRoman}
{
\begin{list}{\Roman{cnt1})}
{
\usecounter{cnt1}
\leftmargin 3mm
\setlength{\leftmargin}{\leftmargin}
\topsep 0pt
\parsep 0pt
\itemsep 0pt}
}
\newcommand{\eRoman}{\end{list}}

\newcommand{\balph}
{
\begin{list}{\alph{cnt1})}
{
\usecounter{cnt1}
\leftmargin 3mm
\setlength{\leftmargin}{\leftmargin}
\topsep 0pt
\parsep 0pt
\itemsep 0pt}
}
\newcommand{\ealph}{\end{list}}

\newcommand{\bAlph}
{
\begin{list}{\Alph{cnt1})}
{
\usecounter{cnt1}
\leftmargin 3mm
\setlength{\leftmargin}{\leftmargin}
\topsep 0pt
\parsep 0pt
\itemsep 0pt}
}
\newcommand{\eAlph}{\end{list}}

\newcommand{\bbullet}
{
\begin{list}{$\bullet$}
{
\leftmargin  2mm
\setlength{\leftmargin}{\leftmargin}
\topsep 3pt
\parsep 0pt
\itemsep 2pt}
}
\newcommand{\ebullet}{\end{list}}

\newcommand{\bdash}
{
\begin{list}{-}
{
\leftmargin 4mm
\setlength{\leftmargin}{\leftmargin}
\topsep 0pt
\parsep 0pt
\itemsep 0pt}
}
\newcommand{\edash}{\end{list}}

\begin{document}

\mainmatter  

\title{Distributed Selection in $O \lp \log n \rp$ Time with $O \lp n \log \log n \rp$ Messages}
\titlerunning{Efficient Distributed Selection Algorithm}

%
%
\author{Piotr Berman \and Junichiro Fukuyama}
\authorrunning{}

\institute{Department of Computer Science and Engineering\\ The Pennsylvaina State University\\ \email{\tt \{berman, jxf140\}@cse.psu.edu}
}

%
%

\toctitle{}
\tocauthor{}
\maketitle

\begin{abstract}

We consider the selection problem on a completely connected network of $n$ processors with no shared memory.  Each processor initially holds a given numeric item of $b$ bits allowed to send a message of $\max \lp b, \lg n \rp$ bits to another processor at a time. On such a {\em communication network ${\cal G}$}, we show that the $k$th smallest of the $n$ inputs can be detected in $O \lp \log n \rp$ time with $O \lp n \log \log n \rp$ messages. The possibility of such a parallel algorithm for this {\em distributed $k$-selection problem} has been unknown despite the intensive investigation on many variations of the selection problem carried out since 1970s.
The main trick of our algorithm is to simulate the comparisons and swaps performed by the {\em AKS sorting network}, the $n$-input sorting network of logarithmic depth discovered by Ajtai, Koml{\'o}s and Szemer{\'e}di in 1983.
We also show the universal time lower bound $\lg n$ for many basic data aggregation problems on ${\cal G}$, confirming the asymptotic time optimality of our parallel algorithm.

\keywords{the selection problem, distributed selection, AKS sorting network, communication network, comparator network}
\end{abstract}

\section{Introduction} \label{Introduction}

The classical $k$-selection problem finds the $k$th smallest of given $n$ numeric items. We consider it on a completely connected network ${\cal G}$ of $n$ processors with no shared memory, each holding exactly one item of $b$ bits initially. A processor node in ${\cal G}$ may send a message of $\max \lp b, \lg n \rp$-bits at any parallel step.
As the performance metrics,
we minimize the parallel running time, and/or the total number of messages as the measure of amount of information transmitted on ${\cal G}$.

The parallel selection problem for connected processors with no shared memory has been extensively investigated for various network topologies, cases of how $n$ inputs are distributed, and other constraints \cite{KLW07,KDG03,SSS92,F83}, as well as on the parallel comparison tree (PCT) and parallel random access machine (PRAM) models with shared memories \cite{Textbook2,CHR93,AKSS89}. In this paper  we focus on the above case, calling it the {\em $k$-distributed selection problem} on a {\em communication network ${\cal G}$}.
As suggested in \cite{KLW07}, this case of ${\cal G}$ has become increasingly significant for the contemporary distributed computing applications such as sensor networks and distributed hash tables: It is common in their performance analysis to count messages arriving at the designated destinations assuming constant time per delivery (called {\em hops}), rather than count how many times messages are forwarded by processor nodes. The considered network ${\cal G}$ models it well.

The distributed selection problem in this particular network case is one of the long time research topics in parallel algorithms. The results are included in the work such as \cite{KLW07,SSS92,FJ82,SS89}. The following summarizes only a few most closely related to our interest: Let $w$ stand for the number of processors initially holding one or more inputs.
The algorithm Frederickson and Johnson developed in \cite{FJ82} finds the $k$th smallest item with $O \lp w \log \frac{k}{w} \rp$ messages on a completely connected or star-shaped processor network. Santoro et al  \cite{SSS92} discovered in 1992 an algorithm with the expected number of messages bounded by
$O \lp w \log \log \min(k, n-k) + w \log w  \rp$.
The more recent result of \cite{KLW07}  explores a case when the processor network has any diameter $D>1$. It
presents a parallel algorithm with the average time bound $O \lp D \log_D n \rp$ that is asymptotically optimal under some probabilistic assumptions, and one with the deterministic time bound $O \lp D \log_D^2 n \rp$.

The first contribution of this paper is a parallel algorithm for the distributed $k$-selection problem that runs in time $O \lp \log n \rp$ with total $O \lp n \log \log n \rp$ messages on ${\cal G}$.  We will prove the following theorem after formulating the problem.

\begin{theorem} \label{MainClaim}
The $k$th smallest of $n$ inputs can be computed distributedly on a communication network ${\cal G}$ in time $O \lp \log n \rp$ with $O \lp n \log \log n \rp$ messages. \qed
\end{theorem}

\noindent
Such a parallel algorithm has been unknown in the long research history on the selection problem. Satisfying the constraint of total $O \lp n \log \log n \rp$ messages, it improves the parallel time bound $O \lp D \log_D^2 n \rp$ in \cite{KLW07} when $\log^2 D \ll \log n$: On a processor network with diameter $D$, a message can be sent to anywhere forwarded by
$O \lp D \rp$ processors. So the theorem means the parallel time bound $O \lp D \log n \rp \ll D \log_D^2 n$.

Our algorithm simulates the comparisons and swaps performed by the {\em AKS sorting network}, the $n$-input sorting network of $O \lp \log n \rp$ depth discovered by Ajtai, Koml{\'o}s and Szemer{\'e}di in 1983 \cite{AKS}. Known for the difficulty of its performance analysis, the AKS sorting network itself has been a research subject in parallel algorithm. Paterson \cite{Paterson} simplified its construction. Seiferas \cite{Seiferas} further clarified it with the estimate that the depth can be at most $7 \cdot 6.07 \lg n=48.79 \lg n$ layers of  $1/402.15$-halvers. Here an {\em $\ve$-halver} $\lp \ep \in (0, 1) \rp$ is a comparator network of a constant depth to re-order $n$ inputs \cite{Paterson,Textbook}, such that for every $z \in \lbr 0, n/2 \rbr \cap \Z$, the left half of the output includes at most $\ve z$ items among the $z$ largest inputs, and the right half at most $\ve z$ among the $z$ smallest inputs.
In the recent work by Goodrich \cite{Goodrich}, the constant factor of the total $O \lp n \log n \rp$ nodes is significantly reduced although its depth bound is $O \lp n \log n \rp$. The overall research interest on the AKS sorting network has been simpler understanding and reduction of the constant factor of the asymptotic quantities.

In order to design our parallel algorithm for the distributed selection problem, we devise an $n$-input comparator network of depth $O \lp \log \log n \rp$ that reduces strangers to at most $O \lp n \log^{-c} n \rp$ for any given constant $c>0$. Here a {\em stranger} is the $z$th largest input ($1 \le z \le n$) that exists in the left half of the output despite   $z > \frac{1}{2}n$, or in the right half despite $z \le \frac{1}{2}n$. Formally we show:

\begin{theorem} \label{Halver}
For every $c \in \R^+$, and $n \in \Z^+$ that is a sufficiently large power of 6, there exists a comparator network ${\cal H}$ to re-order $n$ inputs satisfying the following two.
\begin{enumerate} [i)]
\item The output includes less than $n \lg^{-c} n$ strangers. and
\item The depth is at most  $1296 cd_{\frac{1}{5}} \lg \lg n$ where $d_{\ve}$ stands for the minimum depth of an $\ve$-halver.
\qed
\end{enumerate}
\end{theorem}

Such ${\cal H}$ will be constructed from given $1/5$-halvers through comparator networks we devise in Section 3 as gadgets. A $1/5$-halver of a depth at most $2 \cdot 7^{21}$ is built in Appendix with the classical result of Gabber and Galil on expander graphs \cite{Galil}.
Consequently the depth of ${\cal H}$ is at most $2592 \cdot 7^{21} c \lg \lg n$. This value would be even far larger if $\ve$-halvers with some smaller $\ve$ were used. Our parallel algorithm for the distributed selection problem mimics this ${\cal H}$ also.

Our third result is the running time lower bound $\lg n$ for many data aggregation problems on ${\cal G}$.
The class of problems having the lower bound is huge including the selection problem and many others. The following statement will be confirmed as a corollary to the theorem we will show in \refsec{LB}.

\medskip

\begin{corollary} \label{cor1}
Any parallel algorithm takes at least $\lg n$ steps in the worst case to compute each of the following three problems on a communication network distributedly: i) the $k$-selection problem on $n$ inputs each of at least $\lc \lg n \rc+1$ bits; ii) the problem of finding the sum of $n$ inputs; iii) the problem of counting numeric items among $n$ inputs, each exceeding a given threshold.  \qed
\end{corollary}

\noindent
By our formulation in \refsec{Def}, the statement assumes that each of the $n$ processor nodes holds one numeric item at time 0.
The corollary shows the asymptotic time optimality of our parallel algorithm as well.

The rest of the paper consists as follows. In Section 2, we define general terminology showing related facts. After devising our gadgets from $\ve$-halvers in Section 3, we will prove \refth{Halver} in Section 4, \refth{MainClaim} in Section 5, and the universal lower bound $\lg n$ in Section 6, followed by concluding remarks in Section 7.

\section{General Terminology and Related Facts} \label{Def}

In this paper, a {\em communication network} ${\cal G}$ means a collection of $n$ processor nodes each two of which are connected, capable of exchanging a message of $\max \lp b, \lg n \rp$ bits at a parallel step where $b \in \Z^+$ is given.  We consider a computational problem $P$ whose input is a set of $n$ numeric items of $b$ bits {\em distributed over ${\cal G}$}, $i.e.$, every processor holds exactly one input item at time 0. A parallel algorithm $A$ is said to {\em compute $P$ on ${\cal G}$ distributedly in time $t$ with $m$ messages}, if all the bits of the computed result from the $n$ inputs are stored at a designated processor node after the $t$th parallel step with total $m$ messages. Messages may be exchanged asynchronously on ${\cal G}$.


A {\em comparator network ${\cal H}$} is a directed acyclic graph of 2-input comparators. For notational convenience, we say that the {\em width} of ${\cal H}$ is the maximum number of nodes of the same depth.  In the standard terminology, a numeric item input to a comparator is called {\em wire}. We may call an ordered set of wires {\em array}, for which the set theoretical notation is used.

The AKS sorting network is a comparator network that sorts $n$ inputs with a depth $\Theta \lp \log n \rp$ and the width at most $n$, performing $O \lp n \log n \rp$ comparisons.

We observe here that the $k$th smallest of $n$ items distributed over ${\cal G}$ can be detected in
$O \lp \log n \rp$ time with $O \lp n \log n \rp$ messages the following way. We design such a parallel algorithm $A$ so that a processor node $v$ in ${\cal G}$ simulates from time $ct$ to $c(t+1)$ ($c:$ a constant, $t \in \N$) a comparator of depth $t$ in the AKS sorting network ${\cal H}$. We mean by ``simulate" that $v$ receives two items from other nodes to send their minimum and/or maximum to anywhere on ${\cal G}$.
Due to the width of the AKS sorting network, and since ${\cal G}$ is completely connected,
$n$ processors in ${\cal G}$ can simulate all the swaps and copying performed by ${\cal H}$.
This way $A$ sorts the $n$ inputs in $O \lp \log n \rp$ steps, then fetches the $k$th smallest.

Our parallel algorithm improves the above so that the number of messages is reduced to $O \lp n \log \log n \rp$.

\section{Devising Gadgets from Expander Graphs for Selection}

In this section, we devise some comparator networks $S$ for the design of our parallel algorithm. Such an $S$  re-orders given distinct $m$ numeric items where $m \in 2 \Z$. For a constant $\ve \in (0, 1)$, consider a bipartite graph $(V_1, V_2, E)$ such that $|V_1|=|V_2|=\frac{1}{2}m$ and
\beeq{Expander}
\varepsilon |\Gamma \lp U \rp|
>
(1- \varepsilon) \min \lp \varepsilon |V_i|, |U| \rp,
\eeq
for every subset $U$ of $V_1$ or $V_2$, where $\Gamma(U)$ denotes the neighbor set of $U$. The sets $V_1$ and $V_2$ represent the left and right halves of the input, respectively, and an edge in $E$ a comparison between the two items.
It is straightforward to check that a comparator network $S$ derived from \refeq{Expander} in the obvious manner is an $\ve$-halver.

An {\em expander}  is a graph $G$ whose {\em expansion ratio} $|\Gamma \lp U \rp| / |U|$ is above a constant greater than 1 for any node set $U$ meeting some condition such as $|U| \le \frac{1}{2}|V(G)|$. The bipartite graph $G= (V_1, V_2, E)$ is an expander in this sense whose explicit construction is given in the aforementioned Gabber-Galil result. We demonstrate in Appendix that it leads to a $1/5$-halver of depth $2 \cdot 7^{21}$, and also an $\ve$-halver of a constant depth for any given $\ve$ can be explicitly constructed. The AKS sorting network sorts numeric items with $\ve$-halvers.

It is well-known in spectral graph theory that the minimum expansion ratio is good if $|U|/ |V(G)|$ is bounded and the second largest eigenvalue $\lambda \lp G \rp$ of $G$ is small. In \cite{LPS86}, a $d$-regular graph called {\em ramanujan graph} is explicitly constructed with the asymptotically optimal bound $\lambda \lp G \rp \le 2 \sqrt{d-1}$, improving the Gabber-Galil result. However not much is known beyond this point regarding the explicit construction of an $\ve$-halver.

\medskip

\begin{figure}
\centering
\includegraphics[width=85mm]{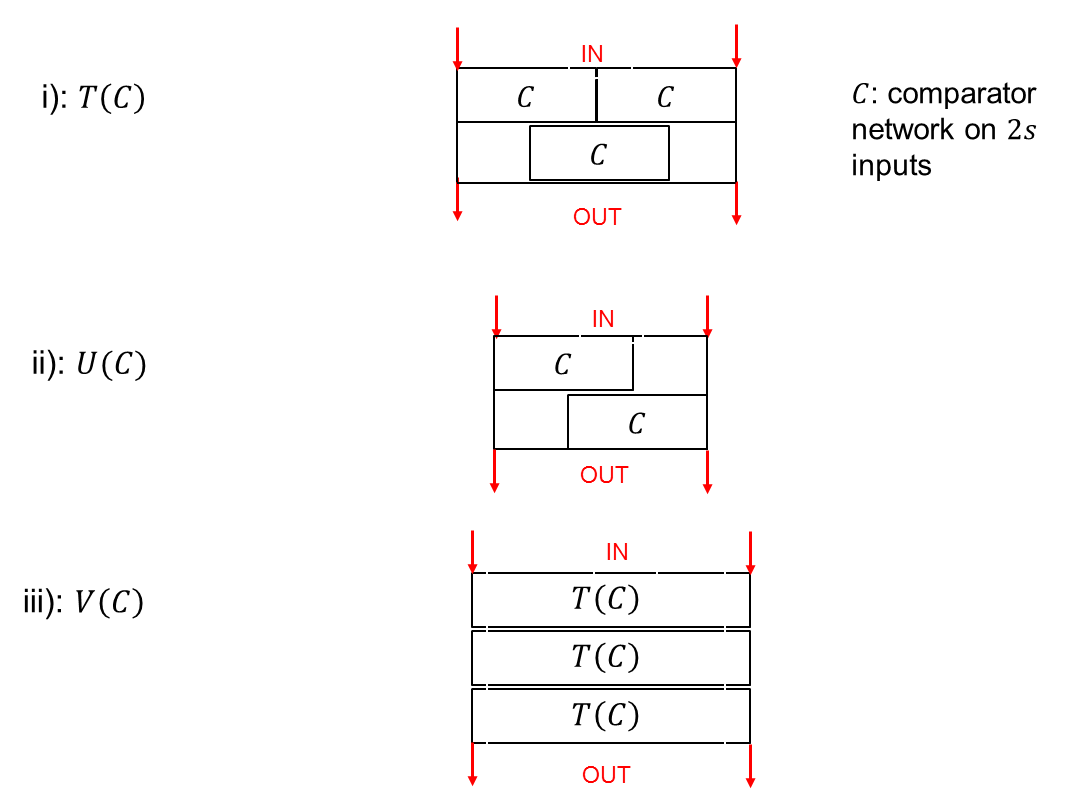}
\caption{Three Comparator Networks $T(C)$, $U(C)$ and $V(C)$} \label{fig1}
\end{figure}

Our gadgets $S$ for the selection problem is constructed from a given $\ve$-halver below.
For any comparator network $C$ with $2s$ inputs and outputs, construct the three comparator networks $T(C)$, $U(C)$ and $V(C)$ as in \reffig{fig1}.
They are on $4s$, $3s$ and $4s$ inputs with 2, 2 and 6 layers of $C$, respectively.
Let $S_0$ be an $\ve$-halver.
Recursively construct
\[
S_i = V \lp S_{i-1} \rp
\eqand
S'_i = U(S_i),
\textrm{~~~for~} i \in \Z^+,
\]
and define
\[
\ve_{i} = \ve_{i-1}^2 + 10 \ve_{i-1}^3,
\]
where $\ve_0 =\ve$.

\begin{figure}
\centering
\includegraphics[width=100mm]{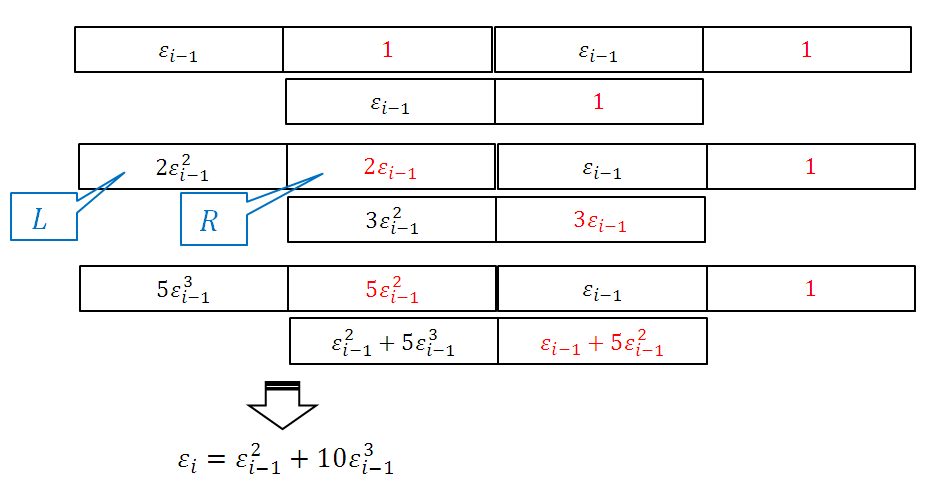}
\caption{Upper Bounds on the $z$-Stranger Ratios in $S_i$} \label{fig1_1}
\end{figure}

Consider the following property $P$ of a comparator network.

\medskip

\noindent
{\bf $P(q, l, s , r )$}  $\Big( q, l \in \Z^+$ and $s, r \in (0,1) \Big)$:
\begin{enumerate} [i)]
\item The comparator network re-orders $m$ inputs ($m \in q \Z^+$) by $l$ layers of $\ve$-halvers.
\item For any given $z \in \Z^+$ at most $\frac{m}{q}$, there are no more than $rz$ items each of rank at most $z$ in the left $s$ of the output.
\end{enumerate}

$S_i$ satisfies the property $P \lp 2^{i+1}, 6^i, \frac{1}{2} , \ve_i \rp$, and $S'_i$ the property \\
$P \lp 3 \cdot 2^i, 2 \cdot 6^i, \frac{2}{3} , 2 \ve_i \rp$.
In \reffig{fig1_1} illustrating $S_i$, consider each $z \le \frac{m}{q}$ where $q=2^{i+1}$.  Let us call an item {\em $z$-stranger} if it has a rank at most $z$, existing in the left $1- \frac{1}{q}$ of the considered layer.
The figure shows the upper bounds on the ratios of $z$-strangers: Each number in black indicates the upper bound on the number of $z$-strangers in the current component divided by $z$.
Each number in red indicates that in the current and its left components divided by $z$.

In particular, there are $2\ve_{i-1} z$ or less $z$-strangers in the component $L \cup R$. As it satisfies the property $P \lp 2^i, 6^{i-1}, \frac{1}{2}, \ve_{i-1} \rp$, the output of $L$ includes $2 \ve_{i-1}^2 z$ or less $z$-strangers. One can check the other bounds to verify the property $P \lp 2^{i+1}, 6^i, \frac{1}{2}, \ve_i \rp$ of $S_i$.
Similarly, $S'_i$ satisfies
$P \lp 3 \cdot 2^i, 2 \cdot 6^i, \frac{2}{3}, 2 \ve_i \rp$.

\begin{figure}
\centering
\includegraphics[width=90mm]{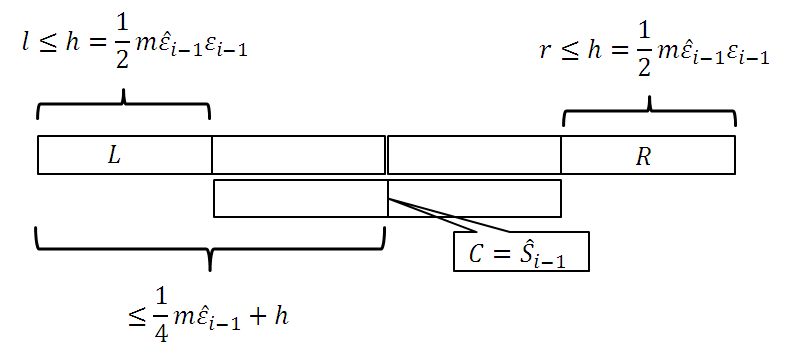}
\caption{To See the Property $\hat{P}$ of $\hat{S}_i$} \label{fig1_2}
\end{figure}

Now let
\[
\ve = \frac{1}{5}.
\]
Construct $\hat{S}_i$ by recursively replacing the second layer of $T \lp S_{i-1} \rp$ by $\hat{S}_{i-1}$, where $\hat{S}_0=S_0$. Define
\[
\hat{\ve}_i = \hat{\ve}_{i-1} \lp \frac{1}{2} +  \ve_{i-1} \rp
\textrm{~for $i=\Z^+$, ~~~and~~~}
\hat{\ve}_0 = \ve.
\]
Here we slightly generalize the definition of a {\em stranger} given in Chapter 1:
The item of rank $z$ is a stranger if $z \le \frac{m}{2}$ and it is in the left half of the considered layer, or if $z > \frac{m}{2}$ and it is in the right half.

Each $\hat{S}_i$ ($1 \le i \le 4$) satisfies $\hat{P} \lp 2^{i+1}, \frac{6^i+4}{5} , \hat{\ve}_{i-1}, \hat{\ve}_i \rp$ where:

\medskip

\noindent
Property $\hat{P} \lp q, l, r, r' \rp$ $\Big( q, l \in \Z^+$ and $r, r'  \in (0, 1) \Big)$:
\begin{enumerate} [i)]
\item The comparator network re-orders $m$ inputs ($m \in q \Z^+$) by $l$ layers of $\ve$-halvers.
\item If there are at most $r m$ strangers in the input, there are at most $r' m$ strangers in the output.
\end{enumerate}

\medskip

\noindent
In \reffig{fig1_2} illustrating $\hat{S}_i$, let $l$ and $r$ be the numbers of strangers output from $L$ and $R$, respectively, so $l \le r \le \frac{1}{2} \hat{\ve}_{i-1} m$.
They are both bounded by  $h=\frac{1}{2}m \hat{\ve}_{i-1} \ve_{i-1}$ since each half of the first layer\footnote{
Here $\frac{1}{2} m \hat{\ve}_{i-1} \le \frac{1}{2} m 2^{-i}$ for $\ve = \frac{1}{5}$ and $i  \le 4$. So we can use the property $P \lp 2^i, 6^{i-1}, \frac{1}{2}, \ve_{i-1} \rp$ of $S_{i-1}$.
}
is $S_{i-1}$. The rank of the median of $m$ inputs is $\frac{1}{4} m - r +l$ in the input of $C$ that is a $\hat{S}_{i-1}$. So the total number of strangers in the left half of the output is bounded by
\[
\frac{1}{4} m \hat{\ve}_{i-1} + (r-l) + l
\le \frac{1}{4} m \hat{\ve}_{i-1} + h
= \frac{1}{2} m \hat{\ve}_i.
\]
This proves the condition ii) of $\hat{P} \lp 2^{i+1}, \frac{6^i+4}{5} , \hat{\ve}_{i-1}, \hat{\ve}_i \rp$. The first condition is straightforward to check.

In addition, construct $\hat{S}$ by vertically aligning $\hat{S}_0, \hat{S}_1, \hat{S}_2, \hat{S}_3$ and $\hat{S}_4$, so that all the $m$ outputs of $\hat{S}_{i-1}$ are input to $\hat{S}_i$. It satisfies
$\hat{P} \lp 32, 315, 1, \hat{\ve} \rp$,
where
\[
\hat{\ve} = \hat{\ve}_4 < 0.02314.
\]
We also denote
\[
S = S_3,~~~
S' = S'_3, \eqand
\ve' = 2 \ve_3 < 0.002644,
\]
so $S$ and $S'$ satisfies $P \lp 16, 216, \frac{1}{2}, \ve' \rp$ and $P \lp 24, 432, \frac{2}{3}, \ve' \rp$, respectively.
We will use the obtained $S$, $S'$ and $\hat{S}$ in the next section.

\section{Reducing Strangers by a Comparator Network of Depth $O \lp \log \log n \rp$}

\begin{figure}
\centering
\includegraphics[width=125mm]{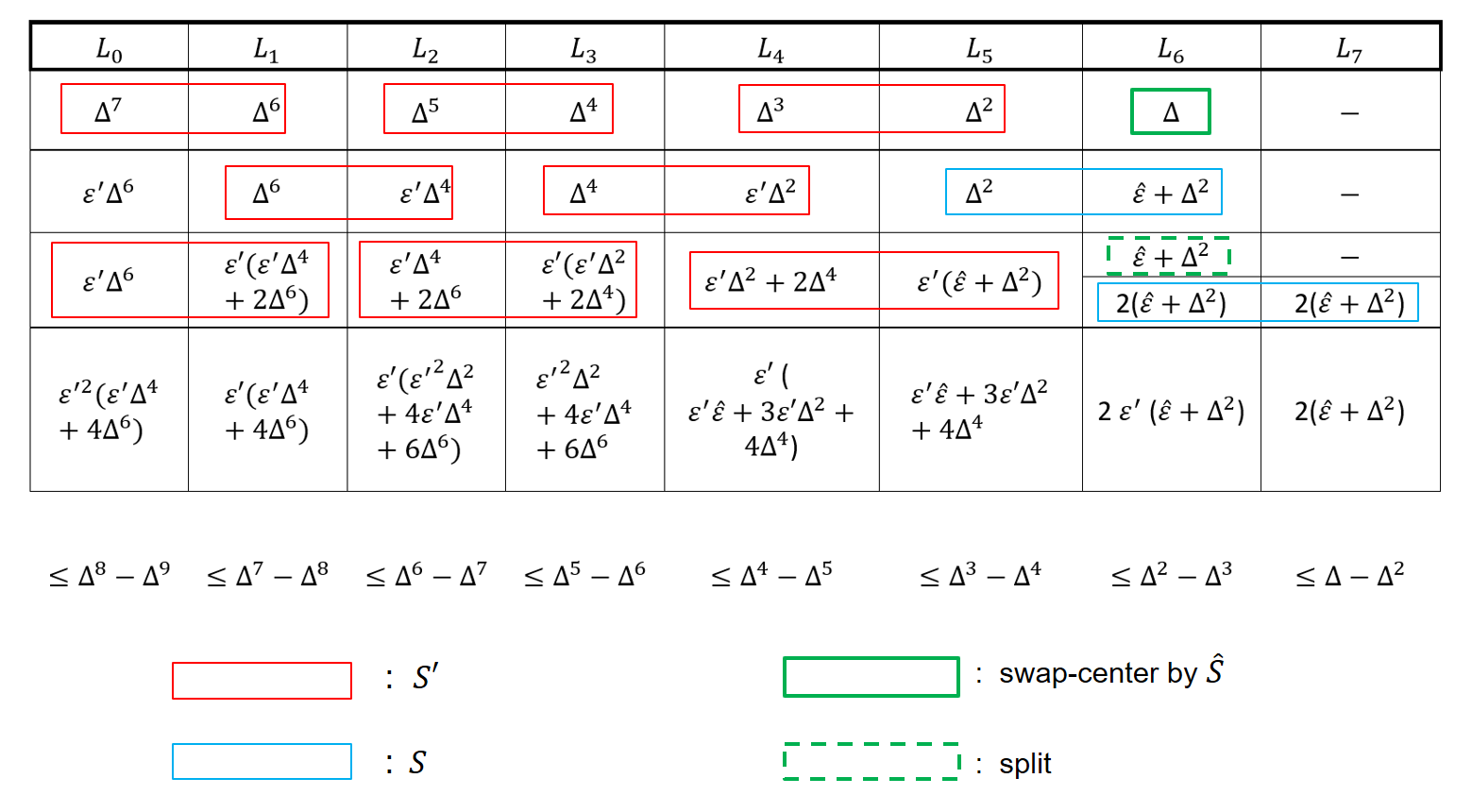}
\caption{Maximum Stranger Ratios in the Segment $j=7$} \label{fig2}
\end{figure}

We prove \refth{Halver} by explicitly building the desired comparator network ${\cal H}$ from given $\ve$-halvers where
$
\ve = 1/5
$.
Assume without loss of generality that the $n$ inputs to ${\cal H}$ have distinct values\footnote{For $s$ elements of a same value $g$, regard that they are valued $g+\delta, g+2 \delta, \ldots, g+ s \delta$ for an infinitesimal number $\delta>0$.
All the statements in our proof remain true with this change.
}.

Here are the construction rules.

\begin{enumerate} [$\bullet$]
\item There are {\em segments} $j=0, 1, 2, \ldots, j_{max}-1$ of ${\cal H}$ where
$j_{max}=\lf c \lg \lg n \rf$. The segment $j>0$ re-orders the $n$ outputs of $j-1$ with a constant number of layers of $\ve$-halvers.

\item The segment $j=0$ consists of $\hat{S}$ on $n$ inputs.

\item
Partition the left half of the input to the $j$th segment into subarrays\\ $L_0, L_1, \ldots, L_{j-1}, L_j$ from the left to right in that order, satisfying
\[
|L_i| =
\lb \begin{array}{cc}
2^{-i-1} n, & \textrm{if $0 \le i \le j-1$,}\\
2^{-i} n, & \textrm{if $i =j$.}\\
\end{array} \right.
\]
So $|L_j|=|L_{j-1}|$.

\item Symmetrically, define $R_0, R_1, \ldots, R_{j-1}, R_j$ in the right half from the right to left.

\item The $j$th segment performs the following four operations.
\begin{enumerate} [-]
\item {\em split:} Split $L_{j-1}$ into its halves re-naming them $L_{j-1}$ and $L_j$. Split $R_{j-1}$ into $R_{j-1}$ and $R_j$ also.
\item {\em compare-even-pairs:} Apply $S'$ or $S$ to every pair $L_{2i}$ and $L_{2i+1}$ ($i=0, 1, \ldots$) currently existing in the left half. Apply $S'$ if they have the different sizes and $S$ otherwise.
Also perform it in the right half symmetrically.
\item {\em compare-odd-pairs:} Same except for pairing $L_{2i+1}$ with $L_{2i+2}$.
\item {\em swap-center:} Apply $\hat{S}$ to $L_{j} \cup R_{j}$.
\end{enumerate}
\item If $j$ is odd, perform {\em compare-even-pairs} and {\em swap-center} simultaneously, then {\em compare-odd-pairs}. Further perform {\em split} and {\em compare-even-pairs}, each once in that order.  If $j$ is even, perform the same operations switching the parity.
\end{enumerate}
\noindent
The segment $j=7$ is illustrated in \reffig{fig2}.

\medskip

Each segment consists of at most $6^4 = 1296$ layers of $\ve$-halvers by the property $P$ of $S$ and $S'$, and $\hat{P}$ of $\hat{S}$.
The total depth of ${\cal H}$ does not exceed $1296  c d_{\frac{1}{5}} \lg \lg n$ as claimed.

Let $\ve'$ and $\hat{\ve}$ be as in the previous section, and let
\[
\Delta = \frac{1}{10}.
\]
Denote by $l_{i, j}$ and $r_{i, j}$ the number of strangers in $L_i$ and $R_i$ of the $j$th segment, respectively.
Our invariant is
\beeq{Invariant}
\sum_{i'=0}^i l_{i', j} \le \Delta^{j-i+1} |L_i|,
\eqand
\sum_{i'=0}^i r_{i', j} \le \Delta^{j-i+1} |R_i|,
\eeq
for every $j$ and $i$ such that $0 \le j < j_{max}$ and $0 \le i \le j$.
We prove it as a lemma.

\begin{lemma} \label{InvariantLemma}
\refeq{Invariant} holds for each $j$ and $i$.
\end{lemma}
\begin{proof}
Show \refeq{Invariant} by induction on $j$.
The $0$th segment consists of $\hat{S}$. The invariant is true by its property $\hat{P} \lp 16, 315, 1, \hat{\ve} \rp$.
Assume true for $j-1$ and prove true for $j$. We only show the induction step for the left half of the $j$th segment when $j$ is odd. The other cases are handled similarly. \reffig{fig2} illustrates the induction step for $j=7$.

Let $l'_i$ be the number of strangers in $L_i$ of the $j$th segment at the considered moment.
It meets the following bound right after {\em compare-even-pairs} and {\em swap-center}.
\[
\frac{l'_i}{|L_i|} \le
\lb \begin{array}{cc}
\ve' \Delta^{j-i-1}, & \textrm{if $i \le j-2$ and $i$ is even,} \\
\Delta^{j-i}, & \textrm{if $i \le j-2$ and $i$ is odd,} \\
\hat{\ve}  + \Delta^2, & \textrm{if $i=j-1$.} \\
\end{array} \right.
\]
The third row of \reffig{fig2} shows the same numbers in the RHS. They mean slightly stronger than the above: The right number in $l$th box  ($l=1, 2, 3$) indicates the upper bound on the number of strangers in $L_{2l-1} \cup L_{2l}$ divided by $|L_{2l}|$.

The correctness of the bounds is due to the properties $P$ and $\hat{P}$ of the used $S'$ and $\hat{S}$. The bound $\hat{\ve}  + \Delta^2$ of the case $i=j-1$ is justified as follows: Its $\Delta^2$ comes from that of the case $i=j-2$ as the median of $L_{j-1} \cup R_{j-1}$ differs from that of all $n$ by at most $\Delta^2 |L_{j-1}|$. Also there are no more than $|L_{j-1}| \hat{\ve}$ elements in $L_{j-1}$ greater than the median of $L_{j-1} \cup R_{j-1}$.

After {\em compare-odd-pairs} and {\em split},
\[
\frac{l'_i}{|L_i|} \le
\lb \begin{array}{cc}
\ve' \Delta^{j-1}, & \textrm{if $i=0$,}\\
\ve' \lp \ve' \Delta^{j-i-2} + 2 \Delta^{j-i} \rp, & \textrm{if $1 \le i \le j-4$ and $i$ is odd,} \\
\ve' \Delta^{j-i-1} + 2 \Delta^{j-i+1},  & \textrm{if $2 \le i \le j-3$ and $i$ is even,} \\
\ve' \lp \hat{\ve} + \Delta^2 \rp, & \textrm{if $i=j-2$,} \\
2 \lp \hat{\ve} + \Delta^2 \rp, & \textrm{if $j-1 \le i \le j$,} \\
\end{array} \right.
\]
as indicated in the 4th row of \reffig{fig2}.

Likewise, right after the 2nd {\em compare-even-pairs}, one can check
\[
\frac{l'_i}{|L_i|} \le \Delta^{j-i+1} - \Delta^{j-i+2},
\]
for each $i \le j$.
These mean the $j$th invariant completing the induction step. The lemma follows.
\qed
\end{proof}

Hence, the total number of strangers in the output of ${\cal H}$ is bounded by
$
2 \Delta |L_{j_{max}-1}| < \Delta 2^{- c \lg \lg n + 3} < n \lg^{-c} n
$. This proves \refth{Halver}.

\section{Selection in $O \lp \log n \rp$ Time with $O \lp n \log \log n \rp$ Messages on a Communication Network}

We show \refth{MainClaim} in this section.
Our parallel algorithm for the distributed $k$-selection problem on a communication network ${\cal G}$ mimics the comparator network ${\cal H}$ given by \refth{Halver} and three AKS sorting networks. Keep assuming that $n$ is a sufficiently large power of 6 and the inputs have distinct values. Let
\[
q=\lg n.
\]
Choose $c=2$ in \refth{Halver}, so the output of ${\cal H}$ includes less than $n \lg^{-2} n$ strangers.
Numeric expressions in this section omit obvious floor/ceiling functions if any.
We first describe the algorithm focusing on the case $k=\frac{n}{2}$.

\medskip

{\noindent\bf Algorithm 1 for Finding the Median of $n$ Inputs:}
\begin{enumerate} [1.]
\item Apply ${\cal H}$ to the given $n$ inputs. Let $L$ and $R$ be its left and right halves of the output, respectively, and $C$ be an empty array.
\item
Arbitrarily partition $L$ into $\frac{n}{2q}$ subarrays each of $q$ items naming them\\ $L_1, L_2, \ldots, L_{\frac{n}{2q}}$.
\item Find the maximum value $l_i$ in each $L_i$.
\item Sort the $l_1, l_2, \ldots, l_{\frac{n}{2q}}$ by an AKS sorting network.
Add the elements of $L_i$ to $C$ for every $i$ such that $l_i$ is among the largest $n \lg^{-2} n$ of $l_1, l_2, \ldots, l_{\frac{n}{2q}}$.

\item Perform Steps 2--4 to $R$ symmetrically to update $C$.

\item Sort the elements of $C$ by another AKS sorting network. Find its median, returning it as the median of all $n$.
\end{enumerate}

Each AKS sorting network used above re-orders at most $2 n \lg^{-1} n$ inputs. This allows us to detect the median with a sufficiently small number of comparisons, leading to $O \lp n \log \log n \rp$ messages. We confirm the correctness of the algorithm first.

\begin{lemma} \label{L3} Algorithm 1 correctly finds the median of all $n$ inputs.
\end{lemma}
\begin{proof}
Let $i$ be an index such that $l_i$ is not among the largest $n \lg^{-2} n$ of $l_1, l_2, \ldots, l_{\frac{n}{2q}}$. If an element of $L_i$ were a stranger or the median of $n$, the largest $n \lg^{-2} n$ items would be all strangers, contradicting \refth{Halver}. So no element of $L - C$ is a stranger or the median. The same is true for $R-C$ except that the median is not included.

Steps 4 and 5 remove the same number of items from the both halves, each neither a stranger nor the median.
Hence Step 6 correctly finds the median of all.
\qed
\end{proof}

We now describe how to run Algorithm 1 on the communication network ${\cal G}$.
Find our implementation below noting two remarks.

\begin{enumerate} [-]
\item Regard that there are $2n$ processor nodes in ${\cal G}$ instead of $n$, since any of them at odd and even time slots can play two different roles.
\item Every processor in ${\cal G}$ can help simulate any of the four comparator networks the way mentioned in \refsec{Def}: It can receive two numeric items from other nodes to send their minimum and/or maximum to anywhere in ${\cal G}$.
\end{enumerate}

\noindent{\bf Implementation of Algorithm 1 on ${\cal G}$: }
Simulate two AKS sorting networks in Steps 4 and 5 with extra $n$ processors. Then select the elements of $C$ in $O \lp \log n \rp$ time as follows. Initially the first node of each $L_i$ is notified if $l_i \in C$. Move items in $L \cap C$ by sending messages so that $L \cap C$ is held by the consecutive smallest numbered nodes.

To find the message destinations, construct,
right after $l_i$ are sorted in Step 4, a complete binary tree $T$ of $\frac{n}{2 q}$ nodes whose leaves are the first nodes of all $L_1, L_2, \ldots, L_{\frac{n}{2q}}$. (For simplicity identify the first node with $L_i$ itself.) Started from the leaves, recursively compute the number of $L_i \subset C$ existing in the subtree $T'$ rooted at each node of $T$. Then, started at the root of $T$, compute the number of $L_i \subset C$ existing outside $T'$ to the left: It is recursively sent from the parent of the current node. Pass it to its left child. Add the total number of $L_i \subset C$ under the left child and send the value to the right child.

This way every leaf of $T$ is informed of the number of $L_i \subset C$ to its left. The leaf disseminates the information to all the nodes in the same $L_i$. Each node is now able to compute the message destination so $L \cap C$ is held by the consecutive smallest numbered nodes.

Merge $L \cap C$ with $R \cap C$ similarly. Perform Step 6 with a simulated AKS sorting network on the obtained $C$. This completes the description of our implementation.

\medskip

Algorithm 1 correctly runs on ${\cal G}$ in $O \lp \log n \rp$  time with $O \lp \log \log n \rp$ messages. One can check it with \reflm{L3} noting that:
\begin{enumerate} [-]
\item The simulated ${\cal H}$ runs in $O \lp \log \log n \rp$ time with $O \lp n \log \log n \rp$ messages.
\item Each of the three simulated AKS sorting networks runs in $O \lp \log n \rp$ time with $O \lp n \rp$ messages.
\item It takes $O \lp n \rp$ messages each of less than $\lg n$ bits to construct $C$.
\end{enumerate}
This verifies \refth{MainClaim} for $k=\frac{n}{2}$.

Selecting an item but the median is done similarly. Assume without loss of generality that we want the $k$th smallest item such that $k < \frac{n}{2}$. We can detect it by adding extra $n-2k$ elements valued $-\infty$ to the input array. This changes no asymptotic bounds we showed so far.

We have constructed a parallel algorithm on ${\cal G}$ that distributedly computes the $k$th smallest of the $n$ inputs in $O \lp \log n \rp$ time with $O \lp n \log \log n \rp$ messages.
We now have \refth{MainClaim}.

\section{The Universal Parallel Time Lower Bound $\lg n$ on a Communication Network} \label{LB}

In this section, we show that it takes at least $\lg n$ steps to compute many basic data aggregation problems on a communication network ${\cal G}$ including the distributed selection problem.  Consider a parallel algorithm to compute a function $f(x_1, x_2, \ldots, x_n) \in \lb 0, 1 \rb$ where $x_1, x_2, \ldots, x_n$  are input numeric items of $b$-bits distributed over ${\cal G}$.
Such $f$ is said to be {\em critical everywhere} if there exist $x_1, x_2, \ldots x_n$ such that
\beeq{CriticalEverywhere}
f \lp x_1, x_2, \ldots, x_{i-1}, \hat{x}_i, x_{i+1}, \ldots, x_n \rp
\ne
f \lp x_1, x_2, \ldots, x_{i-1}, x_i, x_{i+1}, \ldots, x_n \rp,
\eeq
for each index $i=1, 2, \ldots, n$ and some $b$-bit numeric item $\hat{x}_i$.

Below we prove the time lower bound $\lg n$ to compute such $f$ as the following theorem.

\begin{theorem} \label{LowerBound}
Let $f$ be a  function on $n$ inputs critical everywhere.
It takes at least $\lg n$ parallel steps in the worst case to compute $f$ on a communication network distributedly.  \qed
\end{theorem}

\noindent
By the theorem, we will have the statement mentioned in \refsec{Introduction}.

\medskip

\noindent
{\bf Corollary \ref{cor1}.}~
Any parallel algorithm takes at least $\lg n$ steps in the worst case to compute each of the following three problems on a communication network distributedly: i) the $k$-selection problem on $n$ inputs each of at least $\lc \lg n \rc + 1$ bits; ii) the problem of finding the sum of $n$ inputs; iii) the problem of counting numeric items among $n$ inputs, each exceeding a given threshold.
\begin{proof}
i): Given such an algorithm on $n$ inputs $x_1, x_2, \ldots, x_n$, let $f(x_1, x_2, \ldots, x_n)$ be the least significant bit of the $k$th smallest input. Such a function $f$ is critical everywhere: Choose $x_i= 2(i-1)$ for $i=1, 2, \ldots, n$. If $k>1$, let $\hat{x}_i = x_k - 1$, otherwise $\hat{x}_i = 1$.
(All of these numbers are $\lc \lg n \rc + 1$-bit integers.)
These satisfy \refeq{CriticalEverywhere}.  By the theorem, the algorithm takes at least $\lg n$ steps to compute such $f$ on a communication network distributedly.

ii), iii): Shown similarly to i) by choosing $f$ as the least significant bit of the sum of $n$ inputs, and that of the number of inputs exceeding a given threshold, respectively.
\qed
\end{proof}

Consider the class of problems to compute on a communication network distributedly such that any particular bit of the computed result is a function critical everywhere. It is very large containing many aggregation problems such as the above three, each with the parallel running time lower bound $\lg n$ on ${\cal G}$.

\medskip

We prove the theorem.  Denote by $v_i$  the $i$th processor node of ${\cal G}$, and by $x_i$ the numeric item held by $v_i$ at time 0.  Also let $A$ be a parallel algorithm to compute $f$ on ${\cal G}$ distributedly.  Since $f$ is critical everywhere, there exists an input set $X= \lb x_1, x_2, \ldots, x_n \rb$ such that \refeq{CriticalEverywhere}.
Assume $f(x_1, x_2, \ldots, x_n)=1$ without loss of generality, written as $f(X)=1$.

We focus on the computation performed by $A$ on the input set $X$. For time $t \ge 0$ and index $i=1, 2, \ldots, n$, define the set $V_{i, t}$ of nodes in ${\cal G}$ recursively as follows.
\bdash
\item $V_{i, 0} = \lb v_i \rb$ at time $t=0$.
\item $V_{i, t}=V_{i, t-1} \cup V_{j, t-1}$ if  $v_j$ sends a message to $v_i$ at time $t \ge 1$. Otherwise $V_{i, t}=V_{i, t-1}$.
\edash
$V_{i, t}$ is the set of processors $v_j$ such that the values of $x_j$ can possibly affect the computational result at $v_i$ at time $t$.

Let $t_{max}$ be the time when $A$ terminates on $X$, and $v_i$ be the processor that decides $f(X)=1$  at time $t_{max}$.  The size of $V_{i, t_{max}}$ must be $n$; otherwise $A$ decides $f(X)=1$ at $v_i$ with no information on some input $x_j$, contradicting that $f$ is critical everywhere. The following lemma formally confirms it.

\begin{lemma} \label{L2}
$|V_{i, t_{max}}|=n$.
\end{lemma}
\begin{proof}
By the equivalence between an algorithm and Boolean circuit \cite{papa}, there exists a Boolean circuit $C$ to compute $f(X)$,  which is converted from $A$ by the reduction algorithm. Construct the subgraph of $C$ that decides $f(X)=1$ at $v_i$ as follows: For each $t=0, 1, \ldots, t_{max}$, with the algorithm running on $X$ at each processor $v_j$ and the exchanged messages,
inductively construct a Boolean circuit to decide each bit of the computed result stored at $v_j$ at time $t$. For the time $t=t_{max}$ and processor $v_i$, we have a circuit $C'$ that decides $f(X)=1$ only from $x_j$ such that $v_j \in V_{i, t_{max}}$.

The existence of $C'$ means $f(X)=1$ if
all $x_j$ such that $v_j \in V_{i, t_{max}}$ have the values specified by $X$. If $|V_{i, t_{max}}|<n$, this contradicts that $f$ is critical everywhere satisfying \refeq{CriticalEverywhere}. Hence $|V_{i, t_{max}}|=n$.
\qed
\end{proof}

We also have $|V_{j, t}| \le 2^t$ for each $t$ and $j$. It is because $
|V_{j, t}| = \left| V_{j, t-1} \cup  V_{l, t-1} \right|
\le  2 \cdot 2^{t-1}  = 2^t
$ if a processor $v_l$ sends a message to $v_j$ at time $t \ge 1$.

Threfore,
$
\lg n = \lg |V_{i, t_{max}}| \le t_{max}
$.
\refth{LowerBound} follows.

\section{Concluding Remarks and Open Problems}

The depth bound $1296 cd_{\frac{1}{5}} c \lg \lg n$ we showed in \refth{Halver} could be further improved by devising different gadgets using $\ve$-halvers, or different expander graphs. It is interesting to find its limit. It is also possible to improve the aforementioned estimate in \cite{Seiferas} on the constant of the $O \lp \log n \rp$ depth of a sorting network. In addition, it is good to ask if $\Omega \lp \log \log n \rp$ messages are necessary to compute the the distributed selection problem in $O \lp \log n \rp$ time on a communication network.

\section*{Appendix: Explicit Construction of a $1/5$-Halver}

An {\em $\lp m, q, d \rp$-expander graph} is a bipartite graph $(V_1, V_2, E)$ such that $|V_1|=|V_2|=m$, the maximum degree is at most $q$, and
\[
|\Gamma \lp U \rp| \ge  \lp 1+ d \lp 1- \frac{|U|}{m}  \rp\rp |U|,
\]
for every $U \subseteq V_1$. Gabber and Galil demonstrated \cite{Galil} how to explicitly construct an $\lp \hat{m}^2, 7, \frac{2- \sqrt 3}{2} \rp$-expander graph for every $\hat{m} \in \Z^+$. Below we build a $\frac{1}{5}$-halver from such an expander graph.

For any bipartite graph $S=(V_1, V_2, E)$ such that $|V_1|=|V_2|$, denote by $G= f(S)$ the digraph $(V_1, E')$ such that $(v_i, w_j) \in E$ iff $(v_i, v_j) \in E'$ where $v_i$ is the $i$th vertex of $V_1$ and $w_j$ the $j$th of $V_2$.
Also denote $S = f^{-1}(G)$. Let $G^h$ ($h \in \Z^+$) be the digraph $(V_1, E'')$ such that $(v, w) \in E''$ iff there exists a path of length $h$ from $v$ to $w$ in $G$.

Let $S$ be the bipartite graph from the Gabber-Galil construction, and $G=f(S)$. Construct
\[
S' = f^{-1} \lp G^h \rp, \eqwhere h = 21.
\]

We show that the bipartite graph $S'$ satisfies \refeq{Expander} for $\ve=\frac{1}{5}$ and every subset $U$ of $V_1$. It suffices to verify $|\Gamma \lp U \rp| \ge c(1-\ve)m$ for a vertex set $U$ of $G$ such that
$
|U| \ge c \ve m
$, and any $c \in (0, 1]$.
It is further sufficient to assume $c=1$ in the following argument. Let $U_i$ be the neighbor set of $U$ in $G^{i}$ such that $|U|=\ve m$. With the recurrence
\[
\frac{|U_{i+1}|}{m} \ge \lp 1 + \frac{2 - \sqrt 3}{2} \lp 1- \frac{|U_i|}{m} \rp \rp  \frac{|U_i|}{m},
\]
we calculate
$
\frac{|U_h|}{m} > 0.806>1-\ve$ for $h=21$ and $\ve = \frac{1}{5}$. This confirms that a maximum degree at most $7^h$ is sufficient to achieve \refeq{Expander}.

Switch $V_1$ and $V_2$ and construct $S'' = f^{-1} (f(S)^h)$. Our $\frac{1}{5}$-halver is the undirected bipartite graph $(V_1, V_2, E)$ such that $E$ consists all the edges in $S'$ and $S''$. It has a maximum degree at most $2 \cdot 7^{21}$.

For building the comparator network ${\cal H}$ of \refth{Halver},
we must consider a case when $m$ is a sufficient large power of 6 but not a square. If so, construct an $\ve$-halver from the $\lp \lf \sqrt m  \rf^2, 7, \frac{2-\sqrt 3}{2} \rp$-expander graph.
It is straightforward to check that $h=21$ is still sufficient to achieve $\ve = \frac{1}{5}$.

\end{document}